\documentclass{amsart}
\usepackage{amssymb}
\usepackage{mathptmx}
\usepackage{amsthm}
\usepackage{srcltx}
\newtheorem{theorem}{Theorem}[section]
\newtheorem{lemma}[theorem]{Lemma}
\newtheorem{proposition}[theorem]{Proposition}

\theoremstyle{definition}

\theoremstyle{remark}
\newtheorem{remark}[theorem]{Remark}

\begin{document}

\title[Low Weight Codewords of Reed-Muller Codes]{On Low Weight Codewords of Generalized Affine and Projective Reed-Muller Codes}

\author{S. Ballet}
\address{Institut de Math\'{e}matiques
de Luminy\\ case 930, F13288 Marseille cedex 9\\ France}
\email{stephane.ballet@univmed.fr}
\author{R. Rolland}
\address{Institut de Math\'{e}matiques
de Luminy\\ case 930, F13288 Marseille cedex 9\\ France}
\email{robert.rolland@acrypta.fr}
\date{\today}
\begin{abstract}
We propose new results on low weight codewords of affine and projective generalized Reed-Muller codes. 
In the affine case we prove that if the size of the working finite field is large compared to the degree
of the code,
the low weight codewords are products of affine functions. Then in the general case we study
some types of codewords and prove that they cannot be second, thirds or fourth weight depending on
the hypothesis.
In the projective case the second distance of generalized
Reed-Muller codes is estimated, namely a lower bound and an upper bound of this weight
are given. 

\end{abstract}
\maketitle

\section{Introduction - Notations} 
This paper proposes a study on low weight codewords
of generalized Reed-Muller codes  and projective generalized  Reed-Muller codes of degree $d$,
defined over a finite field ${\mathbb F}_q$, 
called respectively GRM codes and PGRM codes.
It includes a focus on their minimum distances as well as the characterization of the codewords reaching
these weights. It also includes  a study of the second weight, namely the weight which is 
just above the minimal distance. The second weight is also called the next-to-minimum weight.

Determining the low weights of the Reed-Muller codes as well as the low weight codewords are
interesting questions related to various fields. Of course, from the point of view of coding theory, 
knowing something on the weight distribution
of a code, and especially on the low weights is a valuable information. From the point of view of 
algebraic geometry the problem is also related to
the computation of the number of rational points of hypersurfaces and in particular hypersurfaces
that are arrangements of hyperplanes. By means of incidence matrices, Reed-Muller codes are related 
to finite geometry codes (see \cite[5.3 and 5.4]{aske}). From this point of view, 
codewords have a geometrical interpretation
and can benefit from the numerous results in this area. Consequently there is a wide variety of 
concepts that may be involved. 

\medskip

Many results concerning this area are here and 
there in various papers. 
In this situation, a comprehensive overview is needed. This is what we do at first in Section \ref{overview}.

Section \ref{minimalw} is an overview on the minimal distance both in the affine case as in the projective case. 
Concerning PGRM codes, 
the second author characterized in \cite{roll1} the codewords of minimal
weights. But the proof given there is sketched. We give in this Section a more detailed proof.

Then in Section \ref{secondw} we recall some results concerning the second weight
an the codewords  of a GRM code reaching 
the second weight.
These codewords are now known. They were determined in \cite{eric}
and \cite{sbou}) for $1 \leq d \leq \frac{q}{2}$ and in \cite{ledu} for the general case.
It should be noted that these codewords are, as for the minimal codewords, products of affine functions.
Next we give new results on affine low weight codewords 
and we split the study in the three following parts:
\begin{itemize}
 \item in Section \ref{largeq} we give new results on low weight codewords in the case where
$q$ is large compared to $d$. We prove that all the configurations of $d$ distinct hyperplanes
have a weight that is lower than the weight of any hypersurface containing an irreducible (absolutely or not)
component of degree $\geq 2$;
\item in section \ref{lowgene} we study the general case and we compare the second, 
third an fourth weight to the weight of a
word which is irreducible but not absolutely irreducible;
\item in Section \ref{dppq} we study the important case where $d <q$ and we prove that 
under some hypothesis, a word which has a factor  irreducible but
not absolutely irreducible has a weight greater than the third weight or than the fourth weight,
depending on the hypothesis.
\end{itemize}

Next, in Section \ref{secondproj} we determine an upper bound and a lower bound for the 
second weight of a PGRM code.

\section{An overview}\label{overview}
\subsection{Polynomials and homogeneous polynomials}
Let ${\mathbb{F}}_q$ be the finite field with $q$ elements 
and $n \geq 1$ an integer. We denote respectively by ${\mathbb A}^n(q)$ and ${\mathbb P}^n(q)$ the affine space
and the projective space of dimension $n$ over ${\mathbb{F}}_q$. 

Let ${\mathbb{F}}_q[X_1,X_2,\cdots,X_n]$ be the algebra of  polynomials in $n$ variables
over ${\mathbb{F}}_q$. If $f$ is in ${\mathbb{F}}_q[X_1,X_2,\cdots,X_n]$ we denote by
$\deg(f)$ its total degree and by $\deg_{X_i}(f)$ its partial degree
with respect to the variable $X_i$. 

Denote by ${\mathcal F}(q,n)$ the space of functions from ${\mathbb{F}}_q^n$
into ${\mathbb{F}}_q$. It is known that any function in ${\mathcal F}(q,n)$ is a 
polynomial function. More precisely there is a surjective linear map $T$ from
${\mathbb{F}}_q[X_1,X_2,\cdots,X_n]$ onto ${\mathcal F}(q,n)$ mapping 
any polynomial on its associated polynomial function:
$$ 
\begin{array}{cccc}
T~: & {\mathbb{F}}_q[X_1,X_2,\cdots,X_n] & \rightarrow & {\mathcal F}(q,n)\\
    &  f                       & \mapsto & T(f)
 \end{array}
$$
where $T(f)(X)=f(X)$ is the evaluation of the polynomial function $f$ at
the point $X=(X_1,X_2,\cdots,X_n)$.
The map $T$ is not injective and has for kernel the ideal generated by 
the $n$ polynomials $X_i^q - X_i$:
$${\rm Ker}(T)=\left ( X_1^q -X_1, X_2^q -X_2,\cdots, X_n^q-X_n \right).$$
Any element of the quotient  ${\mathbb{F}}_q[X_1,X_2,\cdots,X_n]/{\rm Ker}(T)$ can be
represented by a unique reduced polynomial $f$, namely such that for any variable $X_i$ 
the following holds:
$$\deg_{X_i}(f) \leq q-1.$$
We denote by ${\mathcal RP}(q,n)$ the set of reduced polynomials in $n$ variables over ${\mathbb{F}}_q$.
Then, the map $T$ restricted to ${\mathcal RP}(q,n)$ is one to one, namely each function of
${\mathcal F}(q,n)$ can be uniquely represented by a reduced polynomial in ${\mathcal RP}(q,n)$.

Let $d$ be a positive integer. We denote by ${\mathcal RP}(q,n,d)$ the set of reduced polynomials $P$
such that $\deg(P) \leq d$. Remark that if $d \geq n(q-1)$ the set ${\mathcal RP}(q,n,d)$
is the whole set ${\mathcal RP}(q,n)$. 

Let ${\mathcal H}(q,n+1,d)$ the space of homogeneous polynomials in $n+1$ variables
over ${\mathbb{F}}_q$ with total degree $d$. The decomposition
$$
{\mathbb{F}}_q[X_0,X_1,X_2,\cdots,X_n]=\bigoplus_{d\geq 0} {\mathcal H}(q,n+1,d)
$$
provides ${\mathbb{F}}_q[X_0,X_1,X_2,\cdots,X_n]$ with a graded algebra structure. 
Let ${\mathcal J}_d$ be the subspace of polynomials
$f$ in ${\mathcal H}(q,n+1,d)$
such that $f(X)=0$ for any $X\in {\mathbb F}_q^{n+1}$ and denote by $J$ the homogeneous ideal
$$
{\mathcal J}=\bigoplus_{d\geq 0} {\mathcal J}_d.
$$
It is known (cf. \cite{mero} or \cite{reta}) that the ideal ${\mathcal J}$
is the homogeneous ideal generated by the polynomials $X_i^qX_j-X_iX_j^q$ where $0 \leq i <j \leq n$.

\subsection{Generalized Reed-Muller codes}
Let $d$ be an integer such that $1\leq d < n (q-1)$. 
The generalized Reed-Muller code (GRM code) of order $d$ over ${\mathbb F}_q$
is the following subspace of ${\mathbb{F}}_q^{(q^n)}$:
$$
{\rm RM}_q(d,n)=
\left\{\bigl( f(X) \bigr)_{X\in {\mathbb{F}}_q^n}
~| ~ f \in {\mathbb{F}}_q[X_1,\ldots,X_n] \hbox{ and } {\deg}(f) \leq d\right\}.
$$
It may be remarked that the polynomials $f$ determining this code are viewed 
as polynomial functions. Hence each codeword is associated with a unique 
reduced polynomial in ${\mathcal RP}(q,n,d)$.

Let us denote by $Z_a(f)$ the set of zeros of $f$ (where the index $a$ stands for ``affine'').
From a geometrical point of view $Z_a(f)$ is an affine algebraic 
hypersurface in ${\mathbb{F}}_q^n$ and the number of points $N_a(f)=\# Z_a(f)$
of this hypersurface (the number of zeros of $f$) is connected to the weight $W_a(f)$
of the associated codeword by the following formula:
$$W_a(f)=q^n-N_a(f).$$
The code ${\rm RM}_q(d,n)$ has the following parameters (cf. \cite{kalipe}, \cite[p. 72]{blmu})
(where the index $a$ stands for ``affine code''):
\begin{enumerate}
 \item length  $m_a(q,n,d)=q^n$,
 \item dimension  
$$ k_a(q,n,d)= \sum_{t=0}^d\sum_{j=0}^n (-1)^j 
\left( \begin{array}{c} n\\j \end{array} \right) 
\left( \begin{array}{c} t-jq+n-1\\t-jq \end{array} \right),$$
 \item minimum distance $ W_a^{(1)}(q,n,d)=(q-b) q^{n-a-1},$
where $a$ and $b$ are the quotient and the remainder
in the Euclidean division of $d$ by $q-1$, namely
$d=a (q-1)+b$ and $0 \leq b <q-1$.
\end{enumerate}
We denote by $N_a^{(1)}(q,n,d)$ the maximum number of zeros for a 
non-null polynomial function of degree $\leq d$ where $1 \leq d <n(q-1)$,
namely 
$$N_a^{(1)}(q,n,d)=q^n - W_a^{(1)}(q,n,d)=q^n-(q-b) q^{n-a-1}.$$

\begin{remark}
Be careful not to confuse symbols. With our notations,
the Reed-Muller code of order $d$ has length $m_a(q,n,d)$, dimension $k_a(q,n,d)$
and minimum distance $W_a^{(1)}(q,n,d)$. Namely it is an 
$$\left[m_a(q,n,d),k_a(q,n,d),W_a^{(1)}(q,n,d)\right]-\hbox{code}.$$
The integer $n$ is the number of variables of the polynomials
defining the words and the order $d$ is the
maximum total degree of these polynomials.
\end{remark}

The minimum distance of ${\rm RM}_q(d,n)$ was given by T. Kasami, S. Lin, W. Peterson in \cite{kalipe}.
The words reaching this bound were characterized by P. Delsarte, J. Goethals
and F. MacWilliams in \cite{degoma} and are described in the following theorem:

\begin{theorem}[Delsarthe, Goethals, McWilliams]\label{Del}
The maximum number of
rational points over ${\mathbb{F}}_q$, for an algebraic hypersurface $V$
of degree $d$ in the affine space of dimension $n$
which is not the whole space ${\mathbb{F}}_q^n$ is attained if and only if:
$$V=\left ( \bigcup_{i=1}^{a}\left(\textstyle{\bigcup}_{j=1}^{q-1}V_{i,j}\right) \right )
\left ( \bigcup_{j=1}^b W_j\right ) \hbox{ where } d=a(q-1)+b,
$$
with $0 \leq b <q-1$ and where the $V_{i,j}$ and $W_j$ are $d$ distinct hyperplanes defined
on ${\mathbb{F}}_q$ such that
for each fixed $i$ the $V_{i,j}$ are $q-1$ parallel
hyperplanes, the $W_j$ are $b$ parallel hyperplanes and
the $a+1$ distinct linear forms directing these hyperplanes are 
linearly independent.
\end{theorem}
A simpler proof than the original one is given in \cite{ledu2}.

\subsection{Projective generalized Reed-Muller codes}
The case of projective codes is a bit different, because homogeneous
polynomials do not define in a natural way functions on the projective space. 
Let $d$ be an integer such that $1 \leq d \leq n(q-1)$.
The projective generalized
Reed-Muller code of order $d$ (PGRM code) was introduced by G. Lachaud in \cite{lach1}.
Let $S$ a subset of ${\mathbb{F}}_q^{n+1}$ constituted by one point on each punctured
vector line of ${\mathbb{F}}_q^{n+1}$. Remark that any point of the projective space ${\mathbb P}^n(q)$ 
has a unique coordinate representation by an element of $S$. The projective
Reed-Muller code ${\rm PGRM}_q(n,d)$ of order $d$ over ${\mathbb P}^n(q)$ is constituted by the 
words $(f(X))_{X \in S}$
where $f \in {\mathcal H}(q,n+1,d)$ 
and the null word:
$$
{\rm PGRM}_q(n,d)=\left\{\bigl( f(X) \bigr)_{X\in S}
~| ~ f \in {\mathcal H}(q,n+1,d) \right\} \cup \{(0,\cdots,0)\}.
$$
This code is dependent on the set $S$
chosen to represent the points of ${\mathbb P}^n(q)$. But the main parameters are independent of
this choice. Following \cite{lach1} we can choose
$$S=\cup_{i=0}^{n} S_i,$$
where $S_i=\{(0,\cdots,0,1,X_{i+1}, \cdots, X_n)~|~X_k \in {\mathbb{F}}_q\}$. Subsequently, we
shall adopt this value of $S$ to define the code ${\rm PGRM}_q(n,d)$.

For a homogeneous polynomial $f$ let us denote by $Z_h(f)$ the set of zeros of $f$ in the
projective space ${\mathbb P}^n(q)$ (where the index $h$ stands for ``projective''). 
From a geometrical point of view, an element $f \in {\mathcal H}(q,n+1,d)$ defines a
projective hypersurface  $Z_h(f)$ in the projective space ${\mathbb P}^n(q)$. 
The number $N_h(f)=\# Z_h(f) $ of points
of this projective hypersurface
is connected to the weight $W_h(f)$ of the corresponding codeword by the following relation:
$$W_h(f)=\frac{q^{n+1}-1}{q-1} -N_h(f).$$

The parameters of ${\rm PGRM}_q(n,d)$ are the following (cf. \cite{sore2})
(where the index $h$ stands for ``projective code''):

\begin{enumerate}
 \item length  $m_h(q,n,d)=\frac{q^{n+1}-1}{q-1}$,
 \item dimension  
$$
k_h(q,n,d)= 
\sum _{{t=d~mod~q-1}\atop{{0<t \leq r}}}\Biggl (
\sum _{j=0}^{n+1}(-1)^j
\left( 
\begin{array}{l}
n+1 \\ 
~~j~~ \\ 
\end{array}
\right) \times \\
\left( 
\begin{array}{l}
t-jq+n \\ 
~~t-jq~~ \\ 
\end{array}
\right)
\Biggr ),
$$
\item minimum distance: $W_h^{(1)}(q,n,d)=(q-b)q^{n-a-1}$
where $a$ and $b$ are the quotient and the remainder
in the Euclidean division of $d-1$ by $q-1$, namely
$d-1=a (q-1)+b$ and $0 \leq b <q-1$.
\end{enumerate}
We denote by $N_h^{(1)}(q,n,d)$ the maximum number of zeros for a 
non-null homogeneous polynomial function of degree $d$ where $1 \leq d \leq n(q-1)$,
namely 
$$
N_h^{(1)}(q,n,d)=\frac{q^{n+1}-1}{q-1} - W_h^{(1)}(q,n,d)=\frac{q^{n+1}-1}{q-1}-(q-b) q^{n-a-1}.
$$

\section{Minimal distance and corresponding codewords}\label{minimalw}
\subsection{The affine case: GRM codes}
For the affine case recall that we write the degree $d$ in the following form:
\begin{equation}\label{mdec}
d=a(q-1)+b \quad \hbox{with } 0 \leq b < q-1.
\end{equation}
The minimum distance of a GRM code was given by T. Kasami, S. Lin, W. Peterson in \cite{kalipe}.
The words reaching this bound (i.e. the polynomials reaching the
maximal number of zeros) were characterized by P. Delsarte, J. Goethals
and F. MacWilliams in \cite{degoma}. As indicated in \cite{degoma} the polynomials
reaching this bound can be written:
\begin{equation}\label{pol}
P(X)=w_0 \prod_{i=1}^a \left(\strut 1 -(l_i(X)-w_i)^{q-1}\right) 
\prod_{j=1}^b \left(\strut l_{a+1}(X) -w'_j\right)
\end{equation} 
where $X\in {\mathbb F}_q^n$,
the $w'_j$ in the last $b$ factors are distinct elements of ${\mathbb{F}}_q$, the $w_i$
are arbitrary elements of ${\mathbb{F}}_q$ with $w_0 \neq 0$ and $l_i$ are $a+1$ linearly independent
linear forms on ${\mathbb{F}}_q^n$.

Give here the geometric interpretation of such a polynomial $f$ reaching the maximal number of zeros.
The hypersurface defined by $f$ is the following arrangement of hyperplanes:
\begin{enumerate} 
 \item $a$ blocks of $q-1$ parallel hyperplanes, each of them directed by one of the $a$
first linearly independent linear forms $l_i$,
 \item one block of $b$ parallel hyperplanes directed by $l_{a+1}$.
\end{enumerate}
Such a hypersurface will be called a maximal hypersurface and the associated polynomial
is called a maximal polynomial. The corresponding weight is the minimal weight.

\subsection{The projective case: PGRM codes}
Let us denote respectively by $W_h^{(1)}(q,n,d)$ and $W_h^{(2)}(q,n,d)$ the
first and second weight of the projective Reed-Muller code.

\begin{lemma}\label{lemm1}
Let $d > n(q-1)$. Then for any $N$ such that $0 \leq N \leq \frac{q^{n+1}-1}{q-1}$ there exists
a homogeneous polynomial of degree $d$ in $n+1$ variables having  $N$ zeros in ${\mathbb P}^n(q)$.
In particular $W_h^{(1)}(q,n,d)=1$ and $W_h^{(2)}(q,n,d)=2$.
\end{lemma}

\begin{proof}
let 
$$\omega=(0:0: \cdots :1:\omega_{j+1}:\cdots:\omega_n)$$
be a point in ${\mathbb P}^n(q)$ and
$$
f_{\omega}^d(X)= \\X_j^{d-n(q-1)}\prod_{i=0}^{j-1}\left(X_j^{q-1}-X_i^{q-1}\right)\times 
\prod_{i=j+1}^{n}\left(X_j^{q-1} -\left(X_i - \omega_i X_j \right)^{q-1}\right)
$$
be the indicator-function for $\omega$ (cf. \cite{sore2}). The $\frac{q^{n+1}-1}{q-1}$
polynomial functions $f_{\omega}^d(X)$ are a basis for the space of homogeneous polynomials
of degree $d$. 
Let  $U=\{u_1,u_2,\cdots,u_N\}$ be a set consisting of
$N$ distinct points of ${\mathbb P}^n(q)$.
The function
$$f(X)=\sum_{\omega \notin U}f_{\omega}^d(X)$$
has exactly $N$ zeros, namely the points of $U$.
 \end{proof}

\begin{lemma}
For $n=1$ and $d \leq q-1$ the first and the second weight of the projective Reed-Muller code
are respectively
\begin{equation}\label{wh1}
W_h^{(1)}(q,1,d)=q-d+1.
\end{equation}
\begin{equation}\label{wh2}
W_h^{(2)}(q,1,d)=q-d+2.
\end{equation}
\end{lemma}
\begin{proof}
Let $f$ be a homogeneous polynomial in $2$ variables of degree $d$ where  $2 \leq d \leq q-1$.
We can write
$$f(X_0,X_1)=X_0 g(X_0,X_1)+ \lambda X_1^d.$$
where $g$ is homogeneous of degree $d-1$ and $\lambda \in {\mathbb F}_q$.
Let us choose $f$ such that $\lambda \neq 0$.
If $X_0=0$ then $X_1 =1$. Hence $f$  has no zero for $X_0=0$.
If $X_0=1$ then $f(1,X_1)=g(1,X_1)+\lambda X_1^d$. Hence $f(1,X_1)$ is a polynomial
in one variable of degree $d$. Then it is possible to find $f$ such that
$f(1,X_1)$ has $d$ zeros in ${\mathbb F}_q$. In this case $f(X_0,X_1)$ has $d$ zeros in ${\mathbb P}^1(q)$.

Now let us choose $f$ such that $\lambda=0$. In this case $(0:1)$ is a solution
and for $X_0=1$ we have $f(1,X_1)=g(1,X_1)$. Hence we can choose $f$ such that
$f(1,X_1)=g(1,X_1)$ has $d-1$ zeros in ${\mathbb F}_q$. In this case $f(X_0,X_1)$ has also $d$ zeros.
We conclude that $W_h^{(1)}(q,1,d)=(q+1)-d$.

Remark that as $W_h^{(2)}(q,1,d)> W_h^{(1)}(q,1,d)=q-d+1$ we have 
$W_h^{(2)}(q,1,d)\geq q-d+2$. It is straightforward, using for example 
$$f(X_0,X_1)=X_0 g(X_0,X_1)+ X_1^d$$
where $f(1,X_1)$ has $d-1$ zeros in ${\mathbb F}_q$, to build a function $f(X_0,X_1)$ having 
$d-1$ zeros. We conclude that $W_h^{(2)}(q,1,d)= q-d+2$.
 \end{proof}

In order to describe the minimal distance for the projective case, 
write $d-1= a(q-1)+b$ with $0 \leq b < q-1$.
The minimum distance of a PGRM code was given by J.-P. Serre for $d \leq q$ (cf. \cite{serr}),
and by A. S\o rensen in \cite{sore2} for the general case. 
The polynomials reaching the maximal number of zeros (or defining the minimum weighted codewords)
are given by J.-P. Serre for $d \leq q$ (cf. \cite{serr}) and by the last author (cf. \cite{roll1})
for the general case. Let us give a detailed proof of the following result stated in \cite{roll1}.

\begin{theorem}\label{proj1}
Let $f$ be a homogeneous polynomial in $n+1$ variables
of total degree $d$, with coefficients in ${\mathbb F}_q$,  
which does not vanish on the whole projective space ${\mathbb P} ^{n}(q)$.
Then the following holds: 
\begin{enumerate}
\item \label{pt1}
The number of ${\mathbb F}_q$-rational points $N_h(f)$ of the projective
algebraic set defined by $f$ satisfies the following:
\begin{equation}\label{eq1}
N_h(f) \leq \frac{q^{n+1}-1}{q-1}-W_h^{(1)}(q,n,d)
\end{equation}
where
$$W_h^{(1)}(q,n,d)=
\left \{
\begin{array}{ll}
1  & \hbox{ if } d > n(q-1), \\
(q-b)q^{n-a-1} & \hbox{ if } d \leq n(q-1), 
\end{array}
\right .
$$
with
\[ d-1=a(q-1)+b \hbox{ and } \quad 0 \leq b <q-1.\]
\item \label{pt2} The bound in (\ref{eq1}) is attained.
When $d \leq n(q-1)$, the polynomials $f$ attaining this bound are exactly 
the polynomials defining a hypersurface $V=Z_h(f)$ such that:
$V$ contains a hyperplane $H$ (namely $f$ vanishes on $H$) and $V$
restricted to the affine space ${\mathbb A}^n(q)={\mathbb P}^n(q) \setminus H$ 
is a maximal affine hypersurface of ${\mathbb A}^n(q)$.
\end{enumerate}
\end{theorem}

\begin{proof}
The point (\ref{pt1}) is proved by S\o rensen
in \cite{sore2}.
However, in order to prove at the same time 
the point (\ref{pt2}),
let us rewrite entirely the proof given by S\o rensen of the point (\ref{pt1})
and let us show that one can deduce the result (\ref{pt2}) from this proof.

If $d > n(q-1)$, as $f$ does not vanish on  the whole projective space ${\mathbb P}^{n}(q)$,
then $N_h(f) \leq \frac{q^{n+1}-1}{q-1}-1$.
Lemma \ref{lemm1} proves that this bound is attained. 

\medskip

If $d \leq n(q-1)$ and $V=Z_h(f)$ 
contains a hyperplane $H$,
we can suppose that this hyperplane is given by $X_0=0$,
so that $f=X_0f_1$, where $f_1$ is a homogeneous polynomial
of degree $d-1$. 
The complement of $H$ is the affine space
\[ {\mathbb A}^n(q)=\{ X \in {\mathbb P}^n(q)~|~ X_0=1\}.\]
Let ${\widetilde{f_1}}$ be the polynomial in $n$ variables
obtained from $f_1$ by setting $X_0=1$. This polynomial
is defined on ${\mathbb A}^n(q)$ and does not vanish on the whole affine space ${\mathbb A}^n(q)$. Hence,
using the result of Kasami and al. (\cite{kalipe}), we obtain:
$$N_a({\widetilde{f_1}}) \leq q^n-(q-b)q^{n-a-1},$$
and consequently
$$N_h(f)=\# H +  N_a({\widetilde{f_1}})
\leq \frac{q^n-1}{q-1} + q^n - (q-b)q^{n-a-1},$$
$$N_h(f) \leq \frac{q^{n+1}-1}{q-1} - (q-b)q^{n-a-1},$$
where the symbol $\#$ denotes the cardinal.
The bound is attained if and only if the polynomial $\widetilde{f_1}$
verifies the conditions of maximality given in \cite{degoma}. 

If $d \leq n(q-1)$ and $V=Z_h(f)$ does not contain any hyperplane, we give 
a proof of (\ref{eq1}) by induction on $n$.
If $n=1$ and $d>q-1$ we know by Lemma \ref{lemm1} that the result is true. If $d \leq q-1$ the
homogeneous polynomial $f$ in two variables of degree $d$ can be written:
$$f(X_0,X_1)=aX_1^d + bX_0g(X_0,X_1)$$
where $a\neq 0$ and $b\neq 0$ because $V$ does not contain any hyperplane and where 
$g$ is a non null homogeneous polynomial function of degree $d-1$.
The point at infinity $X_0=0, X_1=1$ of the projective line is not a zero, 
then the only zeros are points 
such that $X_0=1$ and $X_1$ is solution of a polynomial equation in one variable of degree $d$.
Then $N_h(f) \leq d$ and the induction property is verified.

Next suppose that the property is true for $n-1$ and $Z_h(f)$ does not contain
any hyperplane. Then for any hyperplane $H$ we have 
$$\#(Z_h(f) \cap H) \leq \frac{q^n-1}{q-1}-W_h^{(1)}(q,n-1,d),$$
$$\#(H\setminus Z_h(f) \cap H) \geq W_h^{(1)}(q,n-1,d).$$
Let us count the number ${\mathcal N}$ of couple $(M,H)$ where $H$ is a hyperplane
and $M$ a point in $\left({\mathbb P}^n(q)\setminus Z_h(f)\right)\cap H$.
We know that the number of hyperplanes containing  a given point is $\frac{q^n-1}{q-1}$.
Then 
$${\mathcal N}=\frac{q^n-1}{q-1} \#\left({\mathbb P}^n(q)\setminus Z_h(f)\right).$$
This number is also the following sum on the $\frac{q^{n+1}-1}{q-1}$ hyperplanes of 
the space ${\mathbb P}^n(q)$
$${\mathcal N}= \sum_H \# (H\setminus Z_h(f)\cap H)\geq \frac{q^{n+1}-1}{q-1}W_h^{(1)}(q,n-1,d).$$
Then 
$$W_h(f) \geq  \frac{q^{n+1}-1}{q^n-1}W_h^{(1)}(q,n-1,d),$$
$$W_h(f) > q W_h^{(1)}(q,n-1,d).$$
As $d \leq n(q-1)$ we have two cases:
\begin{enumerate}
 \item $d \leq (n-1) (q-1)$ and then $W_h^{(1)}(q,n-1,d)= (q-b)q^{n-a-2}$.
Hence $qW_h^{(1)}(q,n-1,d)=(q-b)q^{n-a-1}=W_h^{(1)}(q,n,d)$. In this case we conclude
$$W_h(f) > W_h^{(1)}(q,n,d),$$
which proves that the the induction property is verified and also that
the bound cannot be reached by a hypersurface which does not contain any hyperplane.
 \item $(n-1) (q-1) < d \leq n(q-1)$ and in this case we have $W_h^{(1)}(q,n-1,d)=1$, $a=n-1$ and
$W_h^{(1)}(q,n,d)=q-b$. Then 
 $$W_h(f) >q W_h^{(1)}(q,n-1,d)=q \geq q-b,$$
$$W_h(f) > W_h^{(1)}(q,n,d),$$
which proves that the the induction property is verified and also that
the bound cannot be reached by a hypersurface which does not contain any hyperplane.

The point (\ref{pt2})  is a consequence of the above reasoning. 
\end{enumerate}
 \end{proof}

\section{Low weight codewords in the affine case}
\subsection{The second weight in the affine case}\label{secondw}
Let us denote by $W_a^{(2)}(q,n,d)$ the second weight of the GRM code $RM_q(d,n)$, 
namely the weight which is just
above the minimum distance. Several simple cases can be easily described. 
If $d=1$, we know that the code has only three weights: $0$, the minimum distance 
$W_a^{(1)}(q,n,1)=q^n-q^{n-1}$ and the second weight $W_a^{(2)}(q,n,1)=q^n$.
For $d=2$ and $q=2$ the weight distribution is  more or less a
consequence of the investigation of quadratic forms done by L. Dickson
in \cite{dick} and was also done by E. Berlekamp and N. Sloane in an 
unpublished paper. For $d=2$ and any $q$ (including $q=2$) 
the weight distribution was given
by R. McEliece
in \cite{mcel}. For $q=2$, for any $n$ and any $d$, the weight distribution is 
known in the range $[W_a^{(1)}(2,n,d),2.5W_a^{(1)}(2,n,d)]$ by a result of Kasami, Tokura, Azumi \cite{katoaz}. 
In particular, the second weight is $W_a^{(2)}(2,n,d)=3\times  2^{n-d-1}$ if $1<d<n-1$ and 
$W_a^{(2)}(2,n,d)=2^{n-d+1}$ if $d=n-1$ or $d=1$.
For $d \geq n (q-1)$ the code ${\rm RM}_q(d,n)$ is trivial, namely it is the whole 
${\mathcal F}(q,d,n)$,
hence any integer $0\leq t \leq q^n$ is a weight. 

The general problem of the second weight was tackled by D. Erickson
in his thesis \cite[1974]{eric} and was partly solved.
Unfortunately this very good piece of work was not published and remained virtually unknown.
Meanwhile several authors became interested in the problem. 
The second weight was first studied by J.-P. Cherdieu and R. Rolland in \cite{chro1}
who proved that when $q>2$ is fixed, for $d<q$ sufficiently small
the second weight is
$$W_a^{(2)}(q,n,d)=q^n - d q^{n-1} + (d-1) q^{n-2}.$$
Their result was improved by A. Sboui 
in \cite{sbou}, who
proved the formula for $d \leq q/2$. The methods in \cite{chro1}
and \cite{sbou} are of a geometric nature by means of which the
codewords reaching this weight were determined. These codewords
are hyperplane arrangements.
Then O. Geil in \cite{geil1}, using Gr\"obner basis
methods, proved the formula for $d <q$. 
Moreover as an application of his method, he gave a new
proof of the Kasami-Lin-Peterson minimum distance formula and determined, when
 $d>(n-1) (q-1)$,
the first $d+1 -(n-1) (q-1)$ weights. In particular for $n=2$ the problem is completely solved,
and this case is particularly important as we shall see later.
Finally, the last author in \cite{roll2}, using a mix of Geil's method and geometrical
considerations found the second weight for all cases except when $d=a(q-1)+1$.
However the Gr\"obner basis method does not determine all the codewords reaching the second weight.

Recently, A. Bruen (\cite{brue3}) exhumed the work of Erickson and completed the proof,
solving the problem of the second weight for Generalized Reed-Muller code.
Describe a little more the result of Erickson. First, in order to present his result
introduce the following notation
used in \cite{eric}:
$s$ and $t$ are integers such that
$$d=s(q-1)+t, \hbox{ with } 0<t \leq q-1.$$

\begin{theorem}\label{w2}
The second weight $W_a^{(2)}(q,n,d)$ is
$$
W_a^{(2)}(q,n,d)=W_a^{(1)}(q,n,d)+c q^{n-s-2}
$$
where $W_a^{(1)}(q,n,d)=(q-t)q^{n-s-1}$ is the minimal distance and $c$ is
$$
c= \left \{
 \begin{array}{lll}
   q   &  \hbox{ if } & s=n-1\\
   t-1 &  \hbox{ if } & s<n-1 \hbox{ and } 1<t\leq \frac{q+1}{2}\\
       &  \hbox{ or } & s<n-1 \hbox{ and } t=q-1 \neq 1\\
   q   &  \hbox{ if } & s=0 \hbox{ and } t=1\\
   q-1 &  \hbox{ if } & q<4, s<n-2 \hbox{ and } t=1\\
   q-1 &  \hbox{ if } & q=3, s=n-2 \hbox{ and } t=1\\
   q   &  \hbox{ if } & q=2, s=n-2 \hbox{ and } t=1\\
   q   &  \hbox{ if } & q\geq 4, 0<s \leq n-2 \hbox{ and } t=1\\
   c_t &  \hbox{ if } & q\geq 4, s \leq n-2 \hbox{ and }  \frac{q+1}{2}<t
 \end{array}
\right .
$$
The number $c_t$ is such that $c_t+(q-t)q$ is the second weight for the code ${\rm RM}_q(2,t)$.
\end{theorem}

It results from the previous theorem that 
if one can compute the second weight for a case where $c=c_t$, the problem 
is completely solved. Alternatively, Erickson conjectured that $c_t=t-1$ and 
reduced this conjecture to a conjecture on blocking sets \cite[Conjecture 4.14 p. 76]{eric}. 
Recently in \cite{brue3}
A. Bruen proved that this conjecture follows from two of his papers \cite{brue1}, \cite{brue2}. 
Then the problem is now solved by \cite{eric}+\cite{brue3}. It is also solved by \cite{eric}+\cite{geil1}
(the important case $n=2$ is completely solved in \cite{geil1} 
and this leads to the conclusion as noted above)
or by \cite{eric}+\cite{roll2} (the cases not solved in \cite{eric} are explicitly resolved in \cite{roll2}).
More precisely

\begin{theorem}
The coefficient $c_t$ used in the previous theorem \ref{w2} is
$$c_t=t-1.$$
\end{theorem}

\begin{remark}\label{secondab}
The values $s$ and $t$ are connected to the values $a$ and $b$ of the formula (\ref{mdec}) in the following way:
$a=s$ and $b=t$ unless $t=q-1$ and in this case $a=s+1$ and $b=0$.
Let us also express the second weight with the 
classical writing (\ref{mdec}) for the Euclidean quotient (cf. \cite{roll2}). Let us define $\gamma$
to be such that  
$$W_a^{(2)}(q,n,d)=W_a^{(1)}(q,n,d)+ \gamma\, q^{n-a-2}.$$
The second weight is given by the following:

\par\noindent \phantom{I}I) ${n=1}$ (and then $q>2$): $$W_a^{(2)}(q,n,d)= q-d+1; \quad \gamma=q;$$
\par\noindent II) ${n \geq 2}$
\par\noindent \quad A) ${d=1}$: $$W_a^{(2)}(q,n,d)=q^n;\quad \gamma=q;$$
\par\noindent \quad B) ${d \geq 2}$
\par\noindent \quad\quad 1) ${q=2}$
\par\noindent \quad\quad\quad a) ${2\leq d<n-1}$: $$W_a^{(2)}(q,n,d)=3\times 2^{n-d-1};\quad \gamma=q=2;$$
\par\noindent \quad\quad\quad b) ${d=n-1}$: $$W_a^{(2)}(q,n,d)=4;\quad \gamma=q^2=4;$$
\par\noindent \quad\quad 2) ${q \geq 3}$
\par\noindent \quad\quad\quad a) ${2 \leq d < q-1}$: $$W_a^{(2)}(q,n,d)= q^n-dq^{n-1}+(d-1)q^{n-2};\quad \gamma=b-1=d-1;$$
\par\noindent \quad\quad\quad b) ${(n-1)(q-1) <d < n(q-1)}$: $$W_a^{(2)}(q,n,d)= q-b+1;\quad \gamma=q;$$
\par\noindent \quad\quad\quad c) ${q-1 \leq d \leq (n-1)(q-1)}$
\par\noindent \quad\quad\quad\quad \phantom{ii}i) ${b=0}$: $$W_a^{(2)}(q,n,d)= 2q^{n-a-1}(q-1);\quad \gamma=q(q-2);$$
\par\noindent \quad\quad\quad\quad \phantom{i}ii) ${b=1}$
\par\noindent \quad\quad\quad\quad\quad $\alpha$) ${q=3}$ $$W_a^{(2)}(3,n,d)=8\times 3^{n-a-2};\quad \gamma=q-1;$$
\par\noindent \quad\quad\quad\quad\quad $\beta$) ${q\geq 4}$: $$W_a^{(2)}(q,n,d)=q^{n-a};\quad \gamma=q;$$
\par\noindent \quad\quad\quad\quad iii) ${2 \leq b <q-1}$: $$W_a^{(2)}(q,n,d)=q^{n-a-2}(q-1)(q-b+1); \quad \gamma=b-1.$$
\end{remark}

Finally let us remark that we now have several approaches, close to each other,
but nevertheless different. The first one 
\cite{eric},\cite{brue3} is mainly based on combinatorics of finite geometries, 
the second one \cite{chro1},\cite{sbou}, \cite{roll2}
is mainly based on geometry and hyperplane arrangements, the third \cite{geil1}, \cite{roll2} is 
mainly based on
polynomial study by means of commutative algebra and Gr\"obner basis.  
All these approaches can be fruitful for the study of similar problems,
in particular for the similar codes 
based on incidence structures, finite geometry and incidence matrices
(see \cite{vdv}, \cite{lastvv1}, \cite{lastvv2}, \cite{lastszvv}).

The polynomials reaching the second weight are known  
(cf. \cite[Theorem 3.13, p. 60]{eric}, \cite{sbou} for $2d\leq q$ and \cite{ledu} for any $d$). 

\subsection{Low weight codewords for large $q$}\label{largeq}
The dimension $n$ of the ambient space 
and the degree $d$ are fixed. We make a study for large values of $q$.
We suppose first that $q > d$.
Let us denote by ${\mathcal L\mathcal W}(q,d,n)$ the set of words $f$ (where $f$ is a reduced polynomial) 
of the Reed-Muller code ${\rm RM}_q(d,n)$
such that the set $Z_a(f)$ of zeros of $f$ is an union of $d$ distinct hyperplanes.

\begin{lemma}\label{mini}
Let $f$ be a reduced polynomial function in ${\mathcal F}(q,n)$ which is in ${\mathcal L\mathcal W}(q,d,n)$.
Then the number $N_a(f)$ of zeros in ${\mathbb F}_q^n$ is such that
\begin{equation}
 N_a(f) \geq dq^{n-1} - \frac{d(d-1)}{2}q^{n-2}.
\end{equation}
\end{lemma}
\begin{proof}
The set $Z_a(f)$ of zeros of $f$ is the union of the $d$ distinct hyperplanes $H_i$.
Then 
$$N_a(f)=\#Z_a(f) \geq  \sum_{i=1}^d \#H_i - \sum_{i\neq j} \#  \left( H_i \cap H_j\right).$$ 
But 
$$ \sum_{i\neq j}\#\left( H_i \cap H_j\right)
= \frac{d(d-1)}{2}q^{n-2}.$$
Then
$$N_a(f) \geq dq^{n-1} - \frac{d(d-1)}{2}q^{n-2}.$$
 \end{proof}

The two following lemmas are useful for the study of
irreducible but not absolutely irreducible polynomial functions. 
The first one is a key lemma which can be found 
in \cite{sore1}.The second one is a slight
modification of \cite[Theorem 2.1]{roll1}.

\begin{lemma}\label{mainlemma}
Let $f$ be a non-zero irreducible but not
absolutely irreducible polynomial over the finite field ${\mathbb F}_q$,
in $n$ variables and of degree $d$. Then one can find
a finite extension ${\mathbb F}_{q'}$ such that there exists a unique polynomial
$g$ absolutely irreducible over the finite field ${\mathbb F}_{q'}$, in
$n$ variables and of degree $d'$, satisfying:
$$f=\prod_{\sigma \in G} g^{\sigma},$$
where $G=Gal({\mathbb F}_{q'}/{\mathbb F}_{q})$ is the Galois group of
${\mathbb F}_{q'}$ over ${\mathbb F}_q$ and
$$Deg(f)=[{\mathbb F}_{q'}:{\mathbb F}_q]Deg(g).$$
\end{lemma}

\begin{lemma}\label{mlem}
Let $f\in{\mathcal RP}(q,n,d)$ be an irreducible but not absolutely irreducible
polynomial of degree $d>1$. Let us set $a$ and $b$ such that $d=a(q-1)+b$ and
$0 \leq b <q-1$. Denote by $u$ a number less than or equal to the smallest prime factor of $d$.
Then the number $N_a(f)$ of zeros of $f$ over ${\mathbb{F}}_q$ satisfies:
\begin{equation}
 N_a(f)< q^n -2q^{n-\left\lfloor{\frac{d}{u(q-1)}}\right\rfloor-1}.
\end{equation}
Moreover if $a=0$ 
\begin{equation}
 N_a(f)< \frac{d}{u}q^{n-1}.
\end{equation}
\end{lemma}

\begin{proof}
Using the lemma \ref{mainlemma} we get:
$$Z_a(f)=\bigcup_{\sigma \in G}Z_a(g^{\sigma}).$$
However all the conjugate polynomials $g^{\sigma}$ have the same
zeros in ${\mathbb F}_q$. Hence $Z_a(f)=Z_a(g)$.\par
Let us denote by $s$ the dimension $[{\mathbb F}_{q'}:{\mathbb F}_q]$
of the vector space ${\mathbb F}_{q'}$ over the field ${\mathbb F}_q$.
We know that:
$$d=Deg(f)=sDeg(g)=sd'.$$
If $(w_1,\cdots,w_s)$ is a basis of ${\mathbb F}_{q'}$ over ${\mathbb F}_q$:
$$g(X)=\sum_{j=1}^s h_j(X)w_j,$$
where $h_j \in {\mathcal R}{\mathcal P}(q,d',n)$ and are not all zero.
Hence,
$$Z_a(f)=\bigcap_{j=1}^sZ_a(h_j).$$
All the non-zero $h_j$ cannot be the same products of
degree one polynomials
(in this case, $g$ would be proportional to a polynomial
over ${\mathbb F}_q$),
so that, by the result of Delsarthe, Goethals,
McWilliams \cite{degoma}, $\# Z_a(f)$ cannot attain the
maximum number of zeros given by the formula of Kasami, Lin, Peterson
(\cite{kalipe}):
$$\# Z_a(f)<q^n-(q-b')q^{n-a'-1}$$
where $d'=a'(q-1)+b'$ and $0 \leq b' <q-1$. 
But  
$a'$ is the integer part of $d'/q-1$, namely:
$$a'=\left\lfloor{\frac{d'}{q-1}}\right\rfloor=\left\lfloor{\frac{d}{s(q-1)}}\right\rfloor.$$
In any case:
$$\# Z_a(f)<q^n-(q-(q-2))q^{n-\left\lfloor{\frac{d}{s(q-1)}}\right\rfloor-1}.$$
 As $s$ divides $d$ we have $u \leq s$ and consequently
$$\# Z_a(f)<q^n-2q^{n-\left\lfloor{\frac{d}{u(q-1)}}\right\rfloor-1}.$$
Now, if $a=0$ then $a'=0$ and we can improve the previous
estimate. In this case we know that $b'=d'=d/s$, so that:
$$\# Z_a(f)<q^n-(q-d/s)q^{n-1}.$$
As $s$ divides $d$ we have $u \leq s$ and consequently
the following inequality holds:
$$\# Z_a(f)<{\frac{d}{s}} q^{n-1} \leq {\frac{d}{u}}q^{n-1}.$$
Let us remark that $2 \leq u$ so that if we replace $u$ by $2$, formulas are still valid.
 \end{proof}

\begin{lemma}\label{irrbutnotabs}
Let $g\in {\mathcal F}(q,n)$ such that $\deg(g) \leq d$.
Suppose that $g=g_1 g_2$ where $g_1$ is an irreducible but not absolutely irreducible
polynomial of degree $d' \geq 2$.
Then
$$N_a(g) < \left(d-\frac{d'}{2}\right)q^{n-1}\leq (d-1)q^{n-1}.$$
\end{lemma}

\begin{proof}
By Lemma (\ref{mlem}) we know that 
$$N_a(g_1) < \frac{d'}{2}q^{n-1}.$$
On the other hand, as $g_2$ is not the zero polynomial,
$$N_a(g_2) \leq (d-d')q^{n-1}.$$
Then
$$N_a(g) \leq N_a(g_1)+N_a(g_2) < \left(d-d'+\frac{d'}{2}\right)q^{n-1}=\left(d-\frac{d'}{2}\right)q^{n-1}.$$
As $d'\geq 2$, we have
$$N_a(g)<(d-1)q^{n-1}.$$
 \end{proof}

\begin{proposition}\label{compare1}
 Let $g\in {\mathcal F}(q,n)$ such that $\deg(g) \leq d$.
Suppose that $g=g_1 g_2$ where $g_1$ is an irreducible but not absolutely irreducible
polynomial of degree $d' \geq 2$.
Then if $q \geq \frac{d(d-1)}{2}$,
for any $f \in {\mathcal L\mathcal W}(q,d,n)$ the following inequality holds:
$$N_a(f) > N_a(g).$$
\end{proposition}
\begin{proof}
 We know by Lemma \ref{mini}  that
$$N_a(f) \geq dq^{n-1} - \frac{d(d-1)}{2}q^{n-2}$$
and by Lemma \ref{irrbutnotabs} that
$$N_a(g) < (d-1)q^{n-1}.$$
Then 
$$N_a(f)- N_a(g) > q^{n-1} - \frac{d(d-1)}{2}q^{n-2}.$$
Hence if 
$$q \geq \frac{d(d-1)}{2},$$
we have
$$N_a(f)- N_a(g) >0.$$
 \end{proof}

\begin{lemma}\label{likeweil}
For any absolutely irreducible polynomial function $h$ in ${\mathcal F}(q,n)$
of degree $\leq d$ the following inequality holds:
$$\left | N_a(h) -q^{n-1} \right | \leq A(d)q^{n-\frac{3}{2}}+B(d)q^{n-2},$$
where 
$$A(d)=\sqrt{2}d^{\frac{5}{2}} \hbox{ and } B(d)=4d^2k^{2^k} \hbox{ with } k=\frac{d(d+1)}{2}.$$
\end{lemma}
\begin{proof}
 See \cite[Theorem 5A, p. 210]{schm}.
 \end{proof}

\begin{lemma}\label{absirr}
Let $g\in {\mathcal F}(q,n)$ such that $\deg(g) \leq d$.
Suppose that $g=g_1 g_2$ where $g_1$ is an absolutely irreducible
polynomial of degree $d' \geq 2$.
Then 
$$N_a(g) \leq (d-1)q^{n-1}+A(d)q^{n-\frac{3}{2}}+B(d)q^{n-2}.$$
\end{lemma}
\begin{proof}
$$N_a(g)\leq N_a(g_1)+N_a(g_2).$$
Lemma \ref{likeweil} gives an upper bound for $N_a(g_1)$
and as $g_2$ is not zero, $g_2$ is bounded by $(d-d')q^{n-1}$.
Then
$$N_a(g)\leq (d-d')q^{n-1} +q^{n-1}+A(d')q^{n-\frac{3}{2}}+B(d')q^{n-2},$$
$$N_a(g)\leq (d+1-d')q^{n-1}+A(d')q^{n-\frac{3}{2}}+B(d')q^{n-2},$$
and as $d'\geq 2$ and $A()$, $B()$ are increasing functions 
$$N_a(g)\leq (d-1)q^{n-1}+A(d)q^{n-\frac{3}{2}}+B(d)q^{n-2}.$$
 \end{proof}

\begin{proposition}\label{compare2}
 Let $g\in {\mathcal F}(q,n)$ such that $\deg(g) \leq d$.
Suppose that $g=g_1 g_2$ where $g_1$ is an absolutely irreducible 
polynomial of degree $d' \geq 2$.
Then if $q > q_0(d)$,
where 
$$q_0(d)=\left (\frac{A(d)+\sqrt{A(d)^2+4C(d)}}{2}\right)^2
\hbox{ with } C(d)=B(d)+\frac{d(d-1)}{2},$$
for any $f \in {\mathcal L\mathcal W}(q,d,n)$ the following inequality holds:
$$N_a(f) > N_a(g).$$
\end{proposition}

\begin{proof}
We know by Lemma \ref{mini} that
$$N_a(f) \geq dq^{n-1} - \frac{d(d-1)}{2}q^{n-2}$$
and by Lemma \ref{absirr} that
$$N_a(g) \leq (d-1)q^{n-1}+A(d)q^{n-\frac{3}{2}}+B(d)q^{n-2}.$$
Then we have
$$N_a(f)- N_a(g) \geq q^{n-1} -A(d)q^{n-\frac{3}{2}} -C(d)q^{n-2},$$
$$N_a(f)- N_a(g) \geq q^{n-2} \left(q-A(d)\sqrt{q} -C(d)\right).$$
As $q-A(d)\sqrt{q} -C(d)$ is a quadratic polynomial in $\sqrt{q}$ we can conclude that
if $q>q_0(d)$ then
$$N_a(f)- N_a(g) >0.$$
 \end{proof}

\begin{theorem}\label{poidsplume}
Let $n\geq 2$ and $d\geq 2$ be integers.
For any prime power $q > q_0(d)$,
for any polynomial function $g$ of degree $\leq d$ which is not the product
of affine factors and for any polynomial function $f$ of degree $d$ which is 
the product of $d$ affine factors $l_i(x)+a_i$ pairwise non-proportional the following holds:
\begin{equation}
N_a(f) > N_a(g).
\end{equation}
\end{theorem}
\begin{proof}
Note that 
$$\frac{d(d-1)}{2} < q_0.$$
Then the result is a consequence of Proposition \ref{compare1} and Proposition \ref{compare2}.
 \end{proof}

\begin{remark}
Theorem \ref{poidsplume} can be also expressed in term of weights of codewords.
If $q > q_0(d)$ then any word in ${\mathcal L\mathcal W}(q,d,n)$ has a weight which is 
strictly lower than any word which is not product of degree one factors.
\end{remark}

\begin{remark}
Let us give as examples of codewords in ${\mathcal L\mathcal W}(q,d,n)$ the codewords
associated to hyperplane arrangements ${\mathcal L}$ defined in \cite[Section 2]{roll2}
in the following way.
Let $d=d_1+d_2+\cdots+d_k$ where
\begin{equation}\label{block}
\left \{ 
\begin{array}{l}
 1 \leq d_1 \leq d_2 \cdots \leq d_k \leq q-1,\\
1 \leq k \leq n.
\end{array}
\right .
\end{equation}
Let us denote by $f_1,f_2,\cdots,f_k$ $k$ linearly independent linear forms on ${\mathbb F}_q^n$ and let us
consider the following hyperplane arrangement: for each $f_i$ we have $d_i$ distinct
parallel hyperplanes defined by 
$$f_i(x)=u_{i,j} \hbox{ with } 1 \leq j \leq d_i.$$
This arrangement of $d$ hyperplanes consists of $k$ blocks of parallel hyperplanes,
the $k$ directions of the blocks being linearly independent.
The corresponding codeword 
$$f(x)=\prod_{i=1}^k \prod_{j=1}^{d_i} (f_i(x)-u_{i,j})$$
is in ${\mathcal L\mathcal W}(q,d,n)$ and has the following number of zeros (see \cite[Theorem 2.1]{roll2}):
$$N_a(f)=q^n-q^{n-k}\prod_{i=1}^k (q-d_i).$$
From the point of view of weight distribution, there is a lot of different values
$W_a(f)$ for different $f$ in this class. For example with $k=2$, all
the different pairs $(d_1,d2)$ with $d_1+d_2=d$ and $d_1\leq d_2$ give different $W_a(f)$.
\end{remark}

\subsection{Low weight codewords in the general case}\label{lowgene}
From \cite{ledu} all the next-to-minimal words are known. So the main interest of the following theorem
is to give an estimate of the distance from some type of codewords to the next-to-minimal ones.
\begin{theorem}\label{nai}
If $f \in {\mathcal RP}(q,n,d)$ is an irreducible polynomial but not absolutely 
irreducible, in $n$ variables over ${\mathbb{F}}_q$,
of degree $d>1$ then the weight $W_a(f)$ of the corresponding codeword in ${\rm RM}_q(n,d)$ is
such that $W_a(f)>W_a^{(2)}(q,n,d)$. Moreover in most case we can determine a strictly
positive lower bound for $W_a(f)-W_a^{(2)}(q,n,d)$ (see the proof for the exact values).
\end{theorem}

\begin{proof}
 By Lemma  \ref{mlem} the weight $W_a(f)$ of the codeword associated to $f$ is such that
\begin{equation}\label{wei}
W_a(f)> 2q^{n-\left\lfloor{\frac{d}{u(q-1)}}\right\rfloor-1}.
\end{equation}
Moreover when $a=0$ the following holds:
\begin{equation}\label{wei0}
W_a(f)>q^n - \frac{d}{u}q^{n-1}.
\end{equation}
In general we shall applied this result with $u=2$ unless we have more
information on $d$ and if we need a more accurate inequality.
In the following we compare for any case $W_a(f)$ to $W_a^{(2)}(q,n,d)$
and we prove that $W_a(f) > W_a^{(2)}(q,n,d)$ and mainly we compute a lower bound 
for $W_a(f) - W_a^{(2)}(q,n,d)$. This lower bound will be useful later.

\medskip

For $n=1$ the result is trivial ($f$ does not have any zero).
We suppose now that $n\geq 2$. Subsequently $a_2$ is defined by:
$$a_2=\left\lfloor{\frac{d}{u(q-1)}}\right\rfloor,$$
with $u=2$
unless we specify another value.

\begin{enumerate}
 \item {The case $q=2$.}
\begin{itemize}
 \item $2 \leq d <n-1$. We know that $W_a^{(2)}(q,n,d)=3\times 2^{n-d-1}$.
As $d \geq 2$, we have $a_2 =\left\lfloor\frac{d}{2(q-1)} \right\rfloor\geq 1$.
If $d$ is even then $2a_2=d$ and the following holds:
$$
W_a^{(2)}(q,n,d)=3\times 2^{n-2a_2-1} \leq 3\times 2^{n-a_2-2}\\
 \leq \frac{3}{4}\times 2^{n-a_2} <\frac{3}{4}W_a(f).
$$ 
If $d$ is odd, then $a_2=\frac{d-1}{2}$ and $d=2a_2+1$.
It follows that $W_a(f)>4\times 2^{n-a_2-2}>3\times 2^{n-2a_2-2}=W_a^{(2)}(q,n,d)$.
 \item $d=n-1$. Then $W_a^{(2)}(q,n,d)=4$. As $d\geq 2$ we conclude that $n\geq 3$ and 
$a_2=\left\lfloor\frac{n-1}{2}\right\rfloor\leq \frac{n-1}{2}$. 
Then 
$$W_a(f) > 2^{n-a_2} \geq 2^{\frac{n+1}{2}}\geq 4=W_a^{(2)}(q,n,d).$$
\end{itemize}

\item {The case $q\geq 3$ and $2\leq d<q$.}
\begin{itemize}
 \item $2 \leq d <q-1$. Here $a=0$. Then $W_a(f)>q^n-\frac{d}{2} q^{n-1}$.
On the other hand we have $W_a^{(2)}(q,n,d)=q^n-dq^{n-1}+(d-1)q^{n-2}$.
Then
$$
W_a(f)-W_a^{(2)}(q,n,d)>\frac{d}{2}q^{n-1}-(d-1)q^{n-2},
$$
$$
W_a(f)-W_a^{(2)}(q,n,d)>q^{n-2}\left(\frac{qd}{2}-d+1\right).
$$
But $q\geq 3$ then $\frac{qd}{2} \geq \frac{3}{2}d$ and
$$
W_a(f)-W_a^{(2)}(q,n,d)>2q^{n-2}.
$$
 \item $d=q-1$. In this case $W_a^{(2)}(q,n,d)=2q^{n-1}-2q^{n-2}$
while $a_2=\left\lfloor \frac{1}{2}\right\rfloor=0$ and $W_a(f)>2q^{n-1}$.
Hence
$$W_a(f)-W_a^{(2)}(q,n,d)>2q^{n-2}.$$ 
\end{itemize}

\item {The case $q\geq 3$ and $(n-1)(q-1)<d<n(q-1)$.}\\
In this case $a_2 <\frac{n}{2}$, $W_a^{(2)}(q,n,d)=(q-b+1)$.
On the other hand, $W_a(f)>2q^{n-a_2-1}$. 
If $n=2$ then $a_2=0$ and $W_a(f)>2q >W_a^{(2)}(q,n,d)$.
If $n=3$ then $a_2=1$ and $W_a(f)>2q^{n-2}\geq 2q> W_a^{(2)}(q,n,d)$. If
$n \geq 4$ then $W_a(f)>q^{\frac{n-2}{2}}\geq 2q > W_a^{(2)}(q,n,d)$.

\item {The case $q\geq 3$ and $q\leq d \leq (n-1)(q-1)$.}\\
 $\bullet$ {$b=0$.} In this case $W_a^{(2)}(q,n,d)=2q^{n-a-1}(q-1)$ and
$a_2=\left\lfloor \frac{a}{2}\right\rfloor$. If $a$ is even then $a=2a_2\geq 1$.
Then $W_a^{(2)}(q,n,d)=2q^{n-2a_2}-2q^{n-2a_2-1}$ and $W_a(f)>2q^{n-a_2-1}$. Hence,
$$W_a(f)-W_a^{(2)}(q,n,d)>2q^{n-2a_2}\left(q^{a_2-1}-1\right)+2q^{n-2a_2-1}.$$
As $q^{a_2-1}-1\geq 0$ we conclude that 
$$W_a(f)-W_a^{(2)}(q,n,d)> 2q^{n-a-1}.$$
If $a$ is odd then $a=2a_2+1$ and $W_a^{(2)}(q,n,d)=2q^{n-2a_2-1}-2q^{n-2a_2-2}$ The following formulas hold:
$$w(f)-W_a^{(2)}(q,n,d)>2q^{n-2a_2-1}\left(q^{a_2}-1\right)+2q^{n-2a_2-2}.$$
As $q^{a_2}-1\geq 0$ we conclude that
$$w(f)-W_a^{(2)}(q,n,d)> 2q^{n-a-1}.$$
$\bullet$ {$b=1$.} 
\begin{itemize}
 \item $q=3$. In this case $d=2a+1$, and consequently the lowest prime factor of $d$
is $\geq 3$. Then we shall take $u=3$ for this case. Hence 
$a_2=\left\lfloor \frac{d}{3(q-1)} \right\rfloor=\left\lfloor \frac{d}{6} \right\rfloor<\frac{d}{6}$,
namely $a_2<\frac{a}{3}+\frac{1}{6}$. Moreover $W_a^{(2)}(q,n,d)=8\times 3^{n-a-2}$ and
$W_a(f)>2\times 3^{n-\frac{a}{3}-\frac{1}{6}-1}$. Then
$$W_a(f)-W_a^{(2)}(q,n,d)> 2\times 3^{n-a-2}\left( 3^{\frac{2a}{3}+\frac{5}{6}} -4\right)$$
and as $a\geq 1$ 
$$
W_a(f)-W_a^{(2)}(q,n,d)> 2\times 3^{n-a-2}\left( 3^{\frac{3}{2}} -4\right) > 2\times 3^{n-a-2}.
$$
 \item $q\geq 4$. We know that $W_a^{(2)}(q,n,d)=q^{n-a}$ and $W_a(f)>2 q^{n-a-1}$.
If $a_2=0$ then 
$$W_a(f)-W_a^{(2)}(q,n,d)>2q^{n-1}-q^{n-a} \geq q^{n-1}.$$
If $a=1$ then $d=q \geq 4$ and $a_2\leq \frac{q}{2(q-1} \leq \frac{2}{3} <1$. Then $a_2=0$.
Hence, if $a_2=1$ then $a \geq 2$. Then $W_a(f)>q^{n-2}$ and
$W_a^{(2)}(q,n,d)\leq q^{n-2}$. We conclude that
$$W_a(f)-W_a^{(2)}(q,n,d)>0.$$
If $a_2 \geq 2$, we know that $a_2=\left\lfloor\frac{a(q-1)+1}{2(q-1)}\right\rfloor$
and then $a_2 \leq \frac{a}{2}+\frac{1}{6}$ or $a>2a_2-\frac{1}{3}$. 
Consequently $W_a^{(2)}(q,n,d)<q^{n-2a_2+\frac{1}{3}}$ while $W_a(f)>2q^{n-a_2-1}$,
hence
$$W_a(f)-W_a^{(2)}(q,n,d)>q^{n-2a_2+\frac{1}{3}}\left(2q^{a_2-\frac{4}{3}}-1 \right)>0.$$
\end{itemize}

$\bullet$  {$2\leq b <q-1$.}
We know that $W_a^{(2)}(q,n,d)=q^{n-a-2}(q-1)(q-b+1)$. 
From the definitions we get the two following inequalities:
$$\frac{d}{q-1}-1<a\leq \frac{d}{q-1},$$
$$\frac{d}{2(q-1)}-1<a_2\leq \frac{d}{2(q-1)},$$
then
$$0 \leq a-2a_2 \leq 1.$$
If $a$ is even then $a=2a_2\geq 2$ and
$$W_a^{(2)}(q,n,d)=q^{n-2a_2-2}(q-1)(q-b+1)<q^{n-2a_2}.$$
Hence:
$$W_a(f)-W_a^{(2)}(q,n,d>2q^{n-a_2-1}-q^{n-2a_2},$$
$$W_a(f)-W_a^{(2)}(q,n,d)> q^{n-2a_2}\left(2q^{a_2-1}-1\right),$$
and as $a_2 \geq 1$ we conclude that
$$W_a(f)-W_a^{(2)}(q,n,d)>q^{n-2a_2}=q^{n-a}.$$
If $a$ is odd, $a=2a_2+1$, $a\geq 1$, $a_2\geq 0$.
Moreover 
$$W_a^{(2)}(q,n,d)=q^{n-2a_2-3}(q-1)(q-b+1)<q^{n-2a_2-1}$$
and
$$W_a(f)>2q^{n-a_2-1}.$$
Then
$$W_a(f)-W_a^{(2)}(q,n,d)>q^{n-2a_2-1}\left(2q^{a_2}-1\right),$$
and as $2q^{a_2}-1\geq 1$ we obtain
$$W_a(f)-W_a^{(2)}(q,n,d)>q^{n-2a_2-1}=q^{n-a}.$$
\end{enumerate}
 \end{proof}

From the computations done in the proof of the previous Theorem 
and examples introduced in \cite{roll2} we can
deduce the following:
\begin{theorem} Suppose that $d$ is such that $d=a(q-1)+b$
with $1 \leq a < n-1$ and $2 \leq b <q-1$ (then $q \geq 4$).
If $f \in {\mathcal RP}(q,n,d)$ is an irreducible polynomial but not absolutely 
irreducible, in $n$ variables over ${\mathbb{F}}_q$,
of degree $d>1$ then the weight $W_a(f)$ of the corresponding codeword in ${\rm RM}_q(n,d)$ is
such that $W_a(f)>W_a^{(4)}(q,n,d)$. 
\end{theorem}

\begin{proof}
Recall that to each hyperplane is associated up to a multiplicative non-zero constant
a affine polynomial. To a hyperplane configuration is associated the product
of these affine polynomials.
Let us consider $T_1$, the type 1 hyperplane configuration,
$T_2$, the type 2 hyperplane configuration 
and $T3$, the type 3 hyperplane configuration given in \cite[Section 2.2]{roll2}.
The following inequalities hold (cf. \cite[Propositions 2.6, 2.8]{roll2}):
$$N_a(T_3) > N_a(T_1) > N_a(T_2).$$
Note that $T_3$ defines codewords which have the second weight.
We have computed in the proof of the previous theorem that
$$W_a(f)-W_a^{(2)}(q,n,d) \geq q^{n-a}.$$
But by \cite[Proposition 2.9]{roll2}
$$W_a(T_2)-W_a(T_3)=W_a(T_2)- W_a^{(2)}(q,n,d) = q^{n-a-2}(q-1).$$
Then
$$W_a(f) > W_a(T_2) > W_a(T_1)> W_a(T_3)=W_a^{(2)}(q,n,d),$$
hence
$$W_a(f)> W_a^{(4)}(q,n,d).$$
 \end{proof}

\subsection{Low weight codeword for the important case $d<q$}\label{dppq}
In this case there are results on the third weight codewords given
by F. Rodier and A. Sboui in \cite{rosb}. They proved that for $q \geq 3d-6$ the three first
weights are given only by some hyperplane arrangement. Moreover they proved
that this is no longer the case for 
$$\frac{q}{2}+\frac{5}{2}\leq d <q,$$
in which case the third weight can be obtained also by some hypersurface containing an irreducible quadric.
In the following we study for $d<q$ the case of an irreducible but not absolutely irreducible factor.

\begin{theorem}
If $f \in {\mathcal RP}(q,n,d)$ is a product of two polynomials $f=g\,.\, h$
such that
\begin{enumerate}
 \item $2 \leq d'=\deg(g) \leq d=\deg(f) <q-1$;
 \item $g$ is irreducible but not absolutely irreducible;
\end{enumerate}
then $W_a(f) > W_a^{(2)}(q,n,d)$.
Moreover if $b\geq 3$ and $q \geq 2d-4$ then $W_a(f) > W_a^{(3)}(q,n,d)$ else
if $b\geq 3$ and $q \geq 2d-3$ then $W_a(f) > W_a^{(4)}(q,n,d)$.
\end{theorem}

\begin{proof}
We know by Lemma \ref{irrbutnotabs} that
$$N_a(f) <(d-1)q^{n-1}.$$
On the other hand,
$$W_a^{(2)}(q,n,d)=q^n -dq^{n-1}+(d-1)q^{n-2}.$$
Then
$$W_a(f)-W_a^{(2)}(q,n,d)>q^{n-1}-(d-1)q^{n-2}>0.$$

Consider now the two following hyperplane configurations $S$ and $T$.
The configuration $S$ is given by two blocks of parallel hyperplanes 
directed by two linearly independent linear forms. The first block contains
$b-2$ parallel hyperplanes and the second block contains $2$ parallel hyperplanes.
The number of points of this configuration is (using for example \cite[Theorem 2.1]{roll2}):
$$N_a(S)=q^n-q^{n-2}(q-d+2)(q-2)=dq^{n-1}-(2d-4)q^{n-2}<q^n-W_a^{(2)}(q,n,d).$$
The configuration $T$ is given by three blocks of parallel hyperplanes
directed by three linearly independent linear forms. The first block contains
$b-2$ parallel hyperplanes, the second block and the third blocks contain a unique hyperplane.
The number of points of this configuration is
$$N_a(T)=dq^{n-1}-(2d-3)q^{n-2}q^{n-3}<N_a(S).$$
If $q \geq 2d-4$, we have $W_a(f)>N_a(S)$. Consequently
$$W_a^{(2)}(q,n,d)<W_a(S)<W_a(f),$$
and then $W_a(f)>W_a^{(3)}(q,n,d)$.
Now if $q \geq 2d-3$, $W_a(f)>N_a(T)$ and consequently 
$$W_a^{(2)}(q,n,d)<W_a(S)<W_a(T)<W_a(f).$$
Then $W_a(f)>W_a^{(4)}(q,n,d)$.
 \end{proof}

\section{The second weight in the projective case}\label{secondproj}

In this section we tackle the unsolved problem of finding the second weight $W_h^{(2)}(q,n,d)$ for
PGRM codes.

\begin{lemma}\label{lemm2}
Let $f$ be a homogeneous polynomial in $n+1$ variables
of total degree $d$, with coefficients in ${\mathbb F}_q$,  
which does not vanish on the whole projective space ${\mathbb P}^{n}(q)$.
If there exists a projective hyperplane $H$ such that 
the affine hypersurface $\left({\mathbb P}^n(q) \setminus H\right)\cap Z_h(f)$
contains an affine hyperplane of the affine space ${\mathbb A}^{n}(q)={\mathbb P}^{n}(q) \setminus H$
then the projective hypersurface $Z_h(f)$ contains a projective hyperplane. Moreover,
if the affine hypersurface $\left({\mathbb P}^n(q) \setminus H\right)\cap Z_h(f)$ is an affine arrangement
of hyperplanes then $Z_h(f)$ is a projective arrangement of hyperplanes.
In particular if $f$ restricted to the affine space 
${\mathbb A}^n(q)$ defines a minimal word or a next-to-minimal word then $Z_h(f)$
is a projective arrangement of hyperplanes.
\end{lemma}

\begin{proof}
Suppose that 
$$f(1,X_1,\cdots ,X_n)= \left(\strut l(X_1, \cdots X_n) - \alpha\right)f_1(X_1,\cdots,X_n)$$
where $l(X_1, \cdots X_n)$ is linear,
then 
$$
f(X_0,X_1, \cdots, X_n)=
\left (\strut l(X_1, \cdots, X_n) -\alpha X_0\right)\widetilde{f_1}(X_0,X_1,\cdots,X_n)
$$
where $\widetilde{f_1}(X_0,X_1,\cdots,X_n)$ is the homogeneous polynomial 
obtained by homogenization of $f_1(X_1,\cdots,X_n)$. We conclude that $f$
defines a hypersurface containing a hyperplane.
 \end{proof}

\begin{lemma}\label{ineq}
 For $n \geq 2$ and $d \geq 2$ the following holds
$$W_h^{(1)}(q,n-1,d)+W_a^{(2)}(q,n,d) \leq W_a^{(2)}(q,n,d-1).$$
\end{lemma}

\begin{proof}
Let us introduce the following notations:
$$d-1=a_{d-1}(q-1)+b_{d-1} \hbox{ with } 0 \leq b_{d-1} \leq q-2,$$
$$d= a_d(q-1)+b_d \hbox{ with } 0 \leq b_{d} \leq q-2.$$
The values $\gamma_{d-1}$ and $\gamma_{d}$
are the the coefficient $\gamma$ which occurs in Remark \ref{secondab}, with respect to $d-1$ and $d$.
Then we have
$$W_h^{(1)}(q,n-1,d)=(q-b_{d-1}) q^{n-a_{d-1}-2},$$
$$W_a^{(2)}(q,n,d)=(q-b_d)q^{n-a_d-1} + \gamma_{d}q^{n-a_d-2},$$
$$W_a^{(2)}(q,n,d-1)=(q-b_{d-1})q^{n-a_{d-1}-1} + \gamma_{d-1}q^{n-a_{d-1}-2}.$$
Denote by $\Delta$ the difference
$$
\Delta= W_a^{(2)}(q,n,d-1)
-\left(\strut(W_h^{(1)}(q,n-1,d)+W_a^{(2)}(q,n,d)\right)
$$
\begin{itemize}
 \item If $0 \leq b_{d-1} \leq q-3$ then $q >2$, $b_d=b_{d-1}+1$
and $a_d=a_{d-1}$. In this case let us denote by $a$ the common value of $a_d$ and $a_{d-1}$. Hence
$$\Delta= q^{n-a-2} \left(\strut b_{d-1}+\gamma_{d-1}-\gamma_d\right).$$
  \begin{itemize}
    \item If $a=n-1$ and $b_{d-1}=0$ then $\gamma_{d-1}=q(q-2)$, $\gamma_{d}=q$ and $\Delta=q^{n-a-1}(q-3)$.
    \item If $a=n-1$ and $b_{d-1}>0$ then $\gamma_{d-1}=\gamma_{d}=q$ and $\Delta=q^{n-a-2}b_{d-1}$.
    \item If $a<n-1$, $b_{d-1}=0$ and $q=3$ then $\gamma_{d-1}=3$, $\gamma_d=2$ and  $\Delta=q^{n-a-1}$.
    \item If $a<n-1$, $b_{d-1}=0$ and $q\geq 4$ then $\gamma_{d-1}=q(q-2)$, $\gamma_d=q$ and $\Delta=q^{n-a-1}(q-3)$.
    \item If $a<n-1$, $b_{d-1}=1$, and $q=3$ then $\gamma_{d-1}=2$, $\gamma_{d}=1$ and $\Delta=2q^{n-a-2}$.
    \item If $a<n-1$, $b_{d-1}=1$, and $q\geq 4$ then $\gamma_{d-1}=q$, $\gamma_{d}=1$ and $\Delta=q^{n-a-1}$.
    \item If $a<n-1$ and $b_{d-1}\geq 2$ then $\gamma_{d-1}-\gamma_{d}=-1$ and $\Delta=q^{n-a-2}(b_{d-1}-1)$.
  \end{itemize}
 \item if $b_{d-1}=q-2$ then $a_d=a_{d-1}+1$ and $b_d=0$. In this case
    \begin{itemize}
      \item If $a_{d-1}=n-1$ then $W_a^{(2)}(q,n,d-1)=3$, $W_a^{(2)}(q,n,d)=2$,
$W_h^{(1)}(q,n-1,d)=1$. Then $\Delta=0$.
      \item If $a_{d-1}<n-1$ then
$$\Delta= 2q^{n-a_{d-1}-1} + \gamma_{d-1}q^{n-a_{d-1}-2}-2q^{n-a_{d-1}-2}-q^{n-a_{d-1}-1}-\gamma_{d}q^{n-a_{d-1}-3},$$
$$\Delta= q^{n-a_{d-1}-2} \left(\strut q-2 + \gamma_{d-1}-\frac{\gamma_{d}}{q}\right).$$
      \begin{itemize}
         \item If $a_{d-1}=n-2$ and $q=2$ then $\gamma_{d-1}=2$, $\gamma_d=4$ and $\Delta=0$.
         \item If $a_{d-1}<n-2$ and $q=2$ then $\gamma_{d-1}=\gamma_d=2$ and $\Delta=q^{n-a_{d-1}-2}$.
         \item If $q=3$ then $\gamma_{d-1}=2$, $\gamma_d=3$ and $\Delta=2\times 3^{n-a_{d-1}-2}$.
         \item If $q\geq 4$ then $\gamma_{d-1}=q-3$, $\gamma_d=q(q-2)$ and $\Delta=q^{n-a_{d-1}-2}(q-3)$.
      \end{itemize}
   \end{itemize}
\end{itemize}
 \end{proof}

\begin{remark}
 In the previous lemma, $\Delta\geq 0$ is zero in the following cases:
\begin{itemize}
\item $q=3$, $a_{d-1}=n-1$ and $b_{d-1}=0$, namely $d=2(n-1)+1$.
\item $q=2$, $a_{d-1}=n-2$, namely $d=n-1$.
\item $a_{d-1}=n-1$, $b_{d-1}=q-2$, namely $d=n(q-1)$.
\end{itemize}
\end{remark}

\begin{theorem}\label{proj2}
Let $W_h^{(2)}(q,n,d)$ be the second weight for 
a homogeneous polynomial $f$ in $n+1$ variables ($n \geq2$)
of total degree $d$ ($2\leq d\leq n(q-1)$), with coefficients in ${\mathbb F}_q$,  
which is not maximal.
Then the following holds:
$$W_h^{(1)}(q,n-1,d)+W_a^{(2)}(q,n,d) \leq W_h^{(2)}(q,n,d) \leq  W_a^{(2)}(q,n,d-1).$$
Moreover
$$W_h^{(2)}(q,n,d) \geq \min\left(W_h^{(1)}(q,n-1,d)+W_a^{(3)}(q,n,d),W_a^{(2)}(q,n,d-1)\right).$$
\end{theorem}

\begin{proof}
Remark first that by Lemma \ref{ineq} 
$$W_h^{(1)}(q,n-1,d)+W_a^{(2)}(q,n,d)\leq  W_a^{(2)}(q,n,d-1).$$
Let $f$ such that $Z_h(f)$ is not maximal.
Suppose first that there is a hyperplane $H$ in $Z_h(f)$.
Then we can suppose that 
$$f(X_0,X_1,\cdots,X_n)=X_0 g(X_0,X_1,\cdots,X_n)$$
where $g$ is a homogeneous polynomial of degree $d-1$.
The function 
$$f_1(X_1,\cdots,X_n)=g(1,X_1,\cdots,X_n)$$ defined on the affine space 
${\mathbb A}^n(q)={\mathbb P}^n(q)\setminus H$ is a polynomial function in $n$ variables
of total degree $d-1$. If it was maximum, by Theorem \ref{proj1}, the function $f$ 
would also be maximum.

Then $\# Z_a(f_1) \leq q^n -W_a^{(2)}(q,n,d-1)$. Hence the following holds:
$$\# Z_h(f) \leq \frac{q^{n}-1}{q-1}+q^n -W_a^{(2)}(q,n,d-1),$$
$$\# Z_h(f) \leq \frac{q^{n+1}-1}{q-1}-W_a^{(2)}(q,n,d-1),$$
and the equality holds if and only if $f_1$ reaches the second weight on the affine space ${\mathbb A}^n(q)$.
This case actually occurs. Hence for such a word, in general we have
$$W_h(f) \geq W_a^{(2)}(q,n,d-1),$$
and as the equality occurs, the following holds for the second distance:
$W_h^{(2)}(q,n,d) \leq W_a^{(2)}(q,n,d-1)$.

Suppose now that there is not any hyperplane in the hypersurface $Z_h(f)$. Let $H$ be a hyperplane
and ${\mathbb A}^n(q)={\mathbb P}^n(q) \setminus H$. Then as $H\cap Z_h(f) \neq H$
$$\# \left(H\cap Z_h(f)\right) \leq \frac{q^{n}-1}{q-1}-W_h^{(1)}(q,n-1,d).$$
We know that the first and second weight of a GRM code are  arrangements of hyperplanes, then by Lemma \ref{lemm2}
$$\# \left(\strut Z_h(f)\cap {\mathbb A}^n(q)\right)\leq q^n - W_a^{(3)}(q,n,d).$$

Now we can write
\begin{eqnarray*}
\# Z_h(f)&\leq&
\frac{q^{n}-1}{q-1}-W_h^{(1)}(q,n-1,d)
+ q^n- W_a^{(3)}(q,n,d)\\
&\leq&\frac{q^{n+1}-1}{q-1}
-\left (\strut W_h^{(1)}(q,n-1,d)+W_a^{(3)}(q,n,d)\right)
\end{eqnarray*}
and consequently
$$W_h(f) \geq W_h^{(1)}(q,n-1,d)+W_a^{(3)}(q,n,d)>W_h^{(1)}(q,n-1,d)+W_a^{(2)}(q,n,d).$$
Then, for the second distance the conclusion of the theorem holds.
 \end{proof}

Unfortunately we don't know the value of $W_a^{(3)}(q,n,d)$ and we don't know if the value of the sum
$W_h^{(1)}(q,n-1,d)+W_a^{(3)}(q,n,d)$ is greater than $W_a^{(2)}(q,n,d-1)$ or 
not. What is the exact value of $W_h^{(2)}(q,n,d)$? This question remains open. 


\begin{thebibliography}{10}

\bibitem{aske}
E.F Assmus and J.D. Key.
\newblock {\em {Designs and their Codes}}, volume 103 of {\em Cambridge Tracts
  in Mathematics}.
\newblock Cambridge University Press, 1992.

\bibitem{blmu}
I.F. Blake and R.C. Mullin.
\newblock {\em {The Mathematical Theory of Coding}}.
\newblock Academic Press, 1975.

\bibitem{brue1}
A.~Bruen.
\newblock {Polynomial Multiplicities over Finite Fields and Intersection Sets}.
\newblock {\em Journal of Combinatorial Theory}, 60(1):19--33, 1992.

\bibitem{brue2}
A.~Bruen.
\newblock {Applications of Finite Fields to Combinatorics and Finite
  Geometries}.
\newblock {\em Acta Applicandae Mathematicae}, 93(1--3), 2006.

\bibitem{brue3}
A.~Bruen.
\newblock {Blocking Sets and Low-Weight Codewords in the Generalized
  Reed-Muller Codes}.
\newblock In A.A. Bruen, D.L. Wehlau, and Canadian~Mathematical Society,
  editors, {\em Error-correcting Codes, Finite Geometries, and Cryptography},
  volume 525 of {\em Contemporary Mathematics}, pages 161--164. American
  Mathematical Society, 2010.

\bibitem{chro1}
J.-P. Cherdieu and R.~Rolland.
\newblock {On the Number of Points of Some Hypersurfaces in
  ${\mathbb{F}}_q^n$}.
\newblock {\em Finite Field and their Applications}, 2:214--224, 1996.

\bibitem{degoma}
P.~Delsarte, J.M. Goethals, and F.J. MacWilliams.
\newblock {On Generalized Reed-Muller Codes and their Relatives}.
\newblock {\em Information and Control}, 16:403--442, 1970.

\bibitem{dick}
L.~Dickson.
\newblock {\em {Linear Groups}}.
\newblock Dover Publications, 1958.

\bibitem{eric}
D.~Erickson.
\newblock {\em {Counting Zeros of Polynomials over Finite Fields}}.
\newblock PhD thesis, Thesis of the California Institute of Technology,
  Pasadena California, 1974.

\bibitem{geil1}
O.~Geil.
\newblock {On the Second Weight of Generalized {R}eed-{M}uller codes}.
\newblock {\em Designs,Codes and Cryptography}, 48(3):323--330, 2008.

\bibitem{kalipe}
T.~Kasami, S.~Lin, and W.~Peterson.
\newblock {New Generalizations of the Reed-Muller Codes Part I: Primitive
  Codes}.
\newblock {\em IEEE Transactions on Information Theory}, IT-14(2):189--199,
  March 1968.

\bibitem{katoaz}
T.~Kasami, N.~Tokura, and S.~Azumi.
\newblock {On the Weight Enumeration of Weights less than $2.5$d of
  {R}eed-{M}uller Codes}.
\newblock {\em Information and Control}, 30(4):380--395, 1976.

\bibitem{lach1}
G.~Lachaud.
\newblock {Projective Reed-Muller Codes}.
\newblock In {\em Coding Theory and Applications}, number 311 in Lecture Notes
  in Computer Science, pages 125--129. Springer-Verlag, 1988.

\bibitem{lastszvv}
M.~Lavrauw, L.~Storme, P.~Sziklai, and G.~Van~de Voorde.
\newblock {An Empty Interval in the Spectrum of Small Weight Codewords in the
  Code from Points and $k$-Spaces in $PG(n,q)$}.
\newblock {\em Journal of Combinatorial Theory}.

\bibitem{lastvv1}
M.~Lavrauw, L.~Storme, and G.~Van~de Voorde.
\newblock {On the Code Generated by the Incidence Matrix of Points and
  Hyperplanes in $PG(n,q)$ and its Dual}.
\newblock {\em Designs, Codes and Cryptography}, 48:231--245, 2008.

\bibitem{lastvv2}
M.~Lavrauw, L.~Storme, and G.~Van~de Voorde.
\newblock {On the Code Generated by the Incidence Matrix of Points and
  $k$-Spaces in $PG(n,q)$ and its Dual}.
\newblock {\em Finite Fields and their Applications}, 14:1020--1038, 2008.

\bibitem{ledu2}
E.~Leducq.
\newblock A new proof of delsarte, goethals and mac williams theorem on minimal
  weight codewords of generalized reed-muller codes.
\newblock {\em Finite Fields and their Applications}, 18(3), 2012.

\bibitem{ledu}
E.~Leducq.
\newblock {Second weight codewords of generalized Reed-Muller codes}.
\newblock {\em arXiv; 1203.5244}, 2012.

\bibitem{mcel}
R.~McEliece.
\newblock {Quadratic Forms over Finite Fields and Second-Order {R}eed-{M}uller
  Codes}.
\newblock Technical report, JPL Space Programs Summary III, 1969.

\bibitem{mero}
D.-J. Mercier and R.~Rolland.
\newblock {Polyn\^omes homog\`enes qui s'annulent sur l'espace projectif
  ${\mathbb P}^m({\mathbb{F}}_q)$}.
\newblock {\em Journal of Pure and Applied Algebra}, 124:227--240, 1998.

\bibitem{reta}
C.~Renter{\'\i}a and H.~Tapia-Recillas.
\newblock {Reed-Muller Codes: An Ideal Theory Approach}.
\newblock {\em Communications in Algebra}, 25(2):401--413, 1997.

\bibitem{rosb}
Fran\c{c}ois Rodier and Adnen Sboui.
\newblock Highest numbers of points of hypersurfaces over finite fields and
  generalized reed--muller codes.
\newblock {\em Finite Fields and Their Applications}, 14(3):816--822, July
  2008.

\bibitem{roll1}
R.~Rolland.
\newblock {Number of Points of Non-Absolutely Irreducible Hypersurfaces}.
\newblock In {\em Algebraic Geometry and its Applications}, volume~5 of {\em
  Number Theory and Its Applications}, pages 481--487. World Scientific, 2008.
\newblock Proceedings of the first SAGA Conference, 7-11 May 2007, Papeete.

\bibitem{roll2}
R.~Rolland.
\newblock {The Second Weight of Generalized Reed-Muller Codes in Most Cases}.
\newblock {\em Cryptography and Communications -- Discrete Structures, Boolean
  Functions and Sequences}, 2(1):19--40, 2010.

\bibitem{sbou}
A.~Sboui.
\newblock {Second Highest Number of Points of Hypersurfaces in
  ${\mathbb{F}}_q^n$}.
\newblock {\em Finite Fields and Their Applications}, 13(3):444--449, July
  2007.

\bibitem{schm}
Schmidt.
\newblock {\em Equations over Finite Fields: An elementary Approach}.
\newblock Number 536 in Lecture Notes in Mathematics. Springer Verlag, Berlin,
  Heidelberg, New York, 1976.

\bibitem{serr}
J.-P. Serre.
\newblock {Lettre \`a M. Tsfasman du 24 Juillet 1989}.
\newblock In {\em Journ\'ees arithm\'etiques de Luminy 17--21 Juillet 1989},
  Ast\'erisque, pages 198--200. Soci\'et\'e Math\'ematique de France, 1991.

\bibitem{sore1}
A.B. S{\o}rensen.
\newblock {A Note on Algorithms Deciding Rationality and Absolutely
  Irreducibility Based on the Number of Rational Solutions}.
\newblock {\em RISC-Linz Series}, 91-37.0, August 1991.

\bibitem{sore2}
A.B. S{\o}rensen.
\newblock {Projective Reed-Muller Codes}.
\newblock {\em Transactions on Information Theory}, IT-37(6):1567--1576, 1991.

\bibitem{vdv}
G.~Van~de Voorde.
\newblock {\em {Blocking Sets in Finite Projective Spaces and Coding Theory}}.
\newblock PhD thesis, Thesis Faculteit Wetenschappen Vakgroep Zuivere Wiskunde
  en Computeralgebra, 2010.

\end{thebibliography}
\end{document}